\documentclass{article}

\usepackage{xspace}

\usepackage{wrapfig}
\usepackage{subfigure}

\usepackage{multirow}
\usepackage{tikz}

\usepackage{amsmath}
\usepackage{amsfonts}
\usepackage{amssymb}
\usepackage{amsthm}

    \usepackage{cite}
    \usepackage{wrapfig}
    \ifdim \pgfversion pt < 2 pt
    \tyxpeout{********************}
    \else
    \usetikzlibrary{calc}
    \fi

\newcommand{\eg}{e.g.}
\newcommand{\ie}{i.e.}

\newtheorem*{claim}{Claim}

\theoremstyle{definition}
\newtheorem{definition}{Definition}

\newtheorem{thm}{Theorem}[section]
\newtheorem{theorem}[thm]{Theorem}

\newtheorem{lemma}[thm]{Lemma}

\title{On the Structure of Weakly Acyclic Games\thanks{This is the revised and expanded version of a paper that appeared in the \emph{Proceedings of SAGT 2010}~\cite{fjs10sagt}.  Partially supported by the DIMACS Special Focus on Communication Security and Information Privacy.}}

\author{Alex Fabrikant\thanks{Google Research, 1600 Amphitheatre Pkwy., Mountain View, CA 94043   \texttt{alexf@cal.berkeley.edu} Supported by a Cisco URP grant and a Princeton University postdoctoral fellowship.}
\and
Aaron D. Jaggard\thanks{Department of Computer Science, Colgate University and DIMACS Center, Rutgers University. 13 Oak Dr., Hamilton, NY 13346.  \texttt{adj@dimacs.rutgers.edu} Partially supported by NSF grants 0751674 and 0753492.}
\and
Michael Schapira\thanks{Department of Computer Science, Princeton University, 35 Olden Street, Princeton, NJ 08540, \texttt{ms7@cs.princeton.edu} Supported by NSF grant 0331548.}}

\date{}

\renewcommand{\SS}{\ensuremath{\mathsf{SS}}}
\newcommand{\USS}{\ensuremath{\mathsf{USS}}}
\newcommand{\SSS}{\ensuremath{\mathsf{SSS}}}

\usepackage{color}
\definecolor{light1}{rgb}{0.5,0.5,0.6}

\newcommand{\F}{\footnotesize}

\newcommand{\onlyshort}[1]{}
\newcommand{\onlylong}[1]{#1}

\begin{document}

\maketitle

\begin{abstract}
The class of \emph{weakly acyclic games}, which includes potential games and dom\-i\-nance-solvable games, captures many practical application domains.  In a weakly acyclic game, from any starting state, there is a sequence of better-response moves that leads to a pure Nash equilibrium; informally, these are games in which natural distributed dynamics, such as better-response dynamics, cannot enter \emph{inescapable oscillations}. We establish a novel link between such games and the existence of pure Nash equilibria in subgames. Specifically, we show that the existence of a \emph{unique} pure Nash equilibrium in every \emph{subgame} implies the weak acyclicity of a game. In contrast, the possible existence of \emph{multiple} pure Nash equilibria in every subgame is insufficient for weak acyclicity in general; here, we also systematically identify the special cases (in terms of the number of players and strategies) for which this is sufficient to guarantee weak acyclicity.
\end{abstract}

\section{Introduction}

In many domains, convergence to a pure Nash equilibrium is a
fundamental problem. In many engineered agent-driven systems that fare
best when steady at a pure Nash equilibrium, convergence to equilibrium is
expected\cite{MAS07,LSZ08} to happen via \emph{better-response} (or \emph{best-response}) \emph{dynamics}: Start at some strategy profile. Players take turns, in some
arbitrary order, with each player making a better response (best response)
to the strategies of the other players, \ie, choosing a strategy that
increases (maximizes) their utility, given the current strategies of
the other players. Repeat this process until no player wants to switch
to a different strategy, at which point we reach a pure Nash equilibrium.

For better-response dynamics to converge to a pure Nash equilibrium
regardless of the initial strategy profile,
a \emph{necessary} condition is that, from every strategy profile, there
exist \emph{some} better-response improvement path (that is, a sequence
of players' better responses) leading from that strategy profile to a
pure Nash equilibrium.  Games for which this property holds are called
``weakly acyclic games''~\cite{PeytonYoung93,Mil96}\footnote{In some
  of the economics literature, the terms ``weak finite-improvement
  path property'' (weak FIP) and ``weak finite best-response path
  property'' (weak FBRP) are also used, for weak acyclicity under better- and
best-response dynamics, respectively.}.
Both potential games~\cite{Ros73,MS96} and dominance-solvable games~\cite{Moulin79} are special cases of weakly acyclic games.

In a game that is not weakly acyclic, there is at least one starting state from which the game is guaranteed
to oscillate indefinitely under better-/best-response
dynamics.  Moreover, the weak acyclicity of a game
implies that natural decentralized dynamics (\eg, randomized
better-/best-response, or no-regret dynamics) are stochastically
guaranteed to reach a pure Nash equilibrium~\cite{MYAS07,PeytonYoung93}. Thus,
weakly acyclic games capture the possibility of reaching pure Nash
equilibria via simple, local, globally-asynchronous interactions
between strategic agents, independently of the starting state. We
assert this is \emph{the} realistic notion of ``convergence'' in most
distributed systems.

\subsection{A Motivating Example}

We now look at an example inspired by interdomain routing that has
this natural form of convergence despite it being, formally, possible
that the network will never converge.  In keeping with results that we
study here, we consider best-response dynamics of a routing model in
which each node can see each other node's current strategy, i.e., its
``next hop'' (the node to which it forwards its data en route to the
destination), as contrasted with models where nodes depending on path
announcements to learn this information.  (Levin et al.~\cite{LSZ08}
formalized routing dynamics in which nodes learn about forwarding
through path announcements.)

\begin{figure}[hbt]
\begin{center}
\vspace{-1em}
\includegraphics[height=1in]{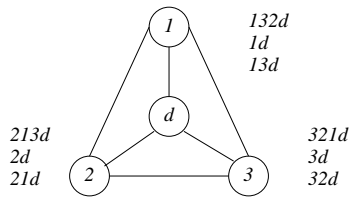}
\caption{Instance of the interdomain routing game that is weakly acyclic and has a best-response cycle.}\label{fig:wa-spp}
\end{center}
\end{figure}

Consider the network on four nodes shown in Fig.~\ref{fig:wa-spp}.
Each of the nodes $1$, $2$, and $3$ is trying to get a path for
network traffic to the destination node $d$.  A strategy of a node $i$
is a choice of a neighbor to whom $i$ will forward traffic; the
strategy space of node $i$, $S_i$, is its neighborhood in the graph.
The utility of the destination $d$ is independent of the outcome, and the utility
$u_i$ of node $i\neq d$ depends only on the path that $i$'s traffic
takes to the destination (and is $-\infty$ if there is no path).  We
only need to consider the relationships between the values of $u_i$ on
all possible paths; the actual values of the utilities do not make a
difference.  Using $132d$ to denote the path from $1$ to $2$ to $3$ to
$d$, and similarly for other paths, here we assume the following:
$u_1(132d) > u_1(1d) > u_1(13d) > -\infty$; $u_2(213d) > u_2(2d) >
u_2(21d) > -\infty$; $u_3(321d) > u_3(3d) > u_3(32d) > -\infty$; and
$u_i(P)=-\infty$ for all other paths $P$, e.g., $u_1(12d) = -\infty$.
These preferences are indicated by the lists of paths in order of
decreasing preference next to the nodes in Fig.~\ref{fig:wa-spp}.

The unique pure Nash equilibrium in the game in Fig.~\ref{fig:wa-spp} is $(d,d,d)$, and, ideally, the dynamics would always converge to it.  However, there exists a best-response cycle in this game as shown in Fig.~\ref{fig:br-cycle}.  Here, each triple lists the paths that nodes $1$, $2$, and $3$ get; the nodes' strategies correspond to the second node in their
respective paths.  The node above the arrow between two triples is the
one that makes a best response to get from one triple to the next.
\begin{figure}[ht]
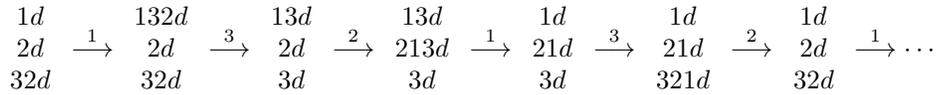

\begin{center}
\[
\begin{array}{c}
  1d \\
  2d \\
  32d
\end{array}
\overset{1}{\longrightarrow}
\begin{array}{c}
  132d \\
  2d \\
  32d
\end{array}
\overset{3}{\longrightarrow}
\begin{array}{c}
  13d \\
  2d \\
  3d
\end{array}
\overset{2}{\longrightarrow}
\begin{array}{c}
  13d \\
  213d \\
  3d
\end{array}
\overset{1}{\longrightarrow}
\begin{array}{c}
  1d \\
  21d \\
  3d
\end{array}
\overset{3}{\longrightarrow}
\begin{array}{c}
  1d \\
  21d \\
  321d
\end{array}
\overset{2}{\longrightarrow}
\begin{array}{c}
  1d \\
  2d \\
  32d
\end{array}
\overset{1}{\longrightarrow}
\cdots
\]
\end{center}
\caption{\label{fig:br-cycle}A best-response cycle for the game in Fig.~\ref{fig:wa-spp}.}
\end{figure}

Once the network is in one of these states,\footnote{For example, this might happen if the link between $2$ and $d$ temporarily fails.  $2$ would always choose to send traffic to $1$ (if anywhere); $1$ would eventually converge to sending traffic directly to $d$ (with $2$ sending its traffic to $1$), and $3$ would then be able to send its traffic along $321d$.  Once the failed link between $2$ and $d$ is restored, $2$'s best response to the choices of the other nodes is to send its traffic directly to $d$, resulting in the first configuration of the cycle above.} there is a fair activation sequence (i.e., in which every node is activated infinitely often) such that each activated node best responds to the then-current choices of the other nodes and such that the network never converges to a stable routing tree (a pure Nash equilibrium).

Although this cycle seems to suggest that the network in
Fig.~\ref{fig:wa-spp} would be operationally troublesome, it is not as
problematic as we might fear.  From every point in the state space,
there is a sequence of best-response moves that leads to the unique
pure Nash equilibrium. We may see this by inspection in this case,
but this example also satisfies the hypotheses of our main theorem below.
So long as each node has some positive
probability of being the next activated node, then, with probability
$1$, the network will eventually converge to the unique stable
routing tree, regardless of the initial configuration of the network.

\subsection{Our Results}

Weak acyclicity is connected to the study of the
computational properties of \emph{sink equilibria} \cite{GMV05,FP08},
minimal collections of states from which best-response dynamics cannot
escape: a game is weakly acyclic if and only if all sinks are
``singletons'', that is, pure Nash equilibria. Unfortunately,
Mirrokni and Skopalik~\cite{MS09} found that reliably checking weak acyclicity
is extremely computationally intractable in the worst case
(PSPACE-complete) even in succinctly-described games. This means, inter alia, that not only can we not
hope to consistently check games in these categories for weak
acyclicity, but we cannot even hope to have general short ``proofs''
of weak acyclicity, which, once somehow found, could be tractably
checked.

With little hope of finding robust, effective ways to consistently
check weak acyclicity, we instead set out to find \emph{sufficient}
conditions for weak acyclicity: finding usable properties that imply
weak acyclicity may yield better insights into at least \emph{some}
cases where we need weak acyclicity for the application.

In this work, we focus on general normal-form games. Potential games,
the much better understood subcategory of weakly acyclic games, are known to have the following property, which we will refer to as
\emph{subgame stability}, abbreviated \SS: not only does a pure Nash
equilibrium exist in the game, but a pure Nash equilibrium exists in
each of its \emph{subgames}, \ie, in each game obtained from the
original game by the removal of players' strategies. Subgame stability
is a useful property in many contexts. For example, in network routing
games, subgame stability corresponds to the important requirement that
there be a stable routing state even in the presence of arbitrary
network malfunctions~\cite{GSW02}. We ask the following natural
question: When is the strong property of subgame stability
\emph{sufficient} for weak acyclicity?

Yamamori and Takahashi~\cite{YT02} prove the following two
results\footnote{Yamamori and Takahashi use the terms \emph{quasi-acyclicity} for weak acyclicity
  under best response, and \emph{Pure Nash Equilibrium Property} (\emph{PNEP})
  for subgame stability.}:

\vspace{0.05in}\noindent{\bf Theorem:}\cite{YT02} \emph{In $2$-player
  games, subgame stability implies weak acyclicity, even under best
  response.}

\vspace{0.05in}\noindent{\bf Theorem:}\cite{YT02} \emph{There exist
  $3 \times 3 \times 3$ games for which subgame stability holds that are
  not weakly acyclic under best response.}\vspace{0.05in}

Thus, subgame stability is sufficient for weak acyclicity in $2$-player games,
yet is not always sufficient for weak acyclicity in games with $n>2$ players. Our goal in this
work is to (1) identify sufficient conditions for weak acyclicity in the general $n$-player case; and (2) pursue a detailed
characterization of the boundary between games for which subgame stability does imply weak acyclicity and games for which it does not.

Our main result for $n$-player games shows that a constraint stronger
than \SS, that we term ``\emph{unique subgame stability}'' (\USS), is sufficient for weak acyclicity:

\vspace{0.05in}\noindent{\bf Theorem:} \emph{If every subgame of a game $\Gamma$ has a unique pure Nash
  equilibrium then $\Gamma$ is weakly acyclic, even under best response.
}\vspace{0.05in}

This result casts an interesting contrast against the negative
result in \cite{YT02}: \emph{unique} equilibria in subgames
guarantee weak acyclicity, but the existence of \emph{more}
pure Nash equilibria in subgames can lead to violations of weak acyclicity. Hence, perhaps counter-intuitively, too many stable states can
potentially result in persistent instability of local dynamics.  (A similar phenomenon is seen in recent work of Jaggard et al.~\cite{jsw11ics}, which studied settings in which multiple stable states preclude the possibility of a non-probabilistic guarantee of convergence.)

We consider $\SS$ games, $\USS$ games, and also the class of
\emph{strict and subgame stable} games $\SSS$, \emph{i.e.}, subgame
stable games which have no ties in the utility functions. We observe
that these three classes of games form the hierarchy $\USS\subset
\SSS\subset \SS$. We examine the number of players, number of
strategies, and the \emph{strictness} of the game (the constraint that
there are no ties in the utility function), and give a complete
characterization of the weak acyclicity implications of each of
these. Our contributions are summarized in Table \ref{table:summary}.

\begin{table}[tb!]
{\renewcommand{\F}{\scriptsize}
\newcommand{\G}{\checkmark}
\newcommand{\B}{\ensuremath{\mathsf{X}}}
\newcommand{\Bweak}{$\not \rightarrow$}
\begin{center}
\begin{tabular}{c|ccccccc|}
\cline{2-8}
& \multicolumn{2}{|c|}{\F 2 players}
& \multicolumn{4}{|c|}{\F 3 players}
& \multicolumn{1}{|c|}{\F 4+ players} \\
\cline{2-8}
& \multicolumn{1}{|c|}{\F $2\times M$}
& \multicolumn{1}{|c|}{\F $ 3 \times M$}
& \multicolumn{1}{|c|}{\F $2\times 2 \times 2$}
& \multicolumn{1}{|c|}{%
\newlength{\BLAH}%
\settowidth{\BLAH}{\F $2\times 3\times M$}%
\begin{minipage}[c]{\BLAH}%
\F $2\times 2 \times M$\\
\F $2\times 3 \times M$\end{minipage}}
& \multicolumn{1}{|c|}{\F $ 2\times  4 \times  4$}
& \multicolumn{1}{|c|}{\F $ 3\times  3 \times  3$}
& \multicolumn{1}{|c|}{\F $ 2\times  2\times  2\times  2$} \\
\hline
%\cline{2-8}
\multicolumn{1}{|c|}{\F $\exists$ pNE}
& \multicolumn{1}{|c|}{\F \G (Lma \ref{lma:2m})}
& \multicolumn{1}{|c|}{\F \B (easy)}
& \multicolumn{1}{|c|}{\F $\G^*$ (Lma \ref{lma:222})}
& \multicolumn{4}{|c|}{\F \B (easy)} \\
\hline
\multicolumn{1}{|c|}{\F \SS}
& \multicolumn{2}{|c|}{\F \G \cite{YT02}}
& \multicolumn{3}{|c|}{\F \B (Thm \ref{thm:counterex-ss})}
& \multicolumn{1}{|c|}{\F \B (Thm \ref{thm:counterex-ss} \& \cite{YT02})}
& \multicolumn{1}{|c|}{\F \B (Thm \ref{thm:counterex-ss})}
\\
\cline{1-1} \cline{4-8}
\multicolumn{1}{|c|}{\F \SSS}
& \multicolumn{2}{|c|}{\F } % implied by YT
& \multicolumn{1}{|c|}{\F \G (Lma \ref{lma:222})}
& \multicolumn{1}{|c|}{\F \G (Thm \ref{thm:sss3p})}
& \multicolumn{1}{|c|}{\F \B (Thm \ref{thm:counterex3p})}
& \multicolumn{1}{|c|}{\F \B (Thm \ref{thm:counterex3p} \& \cite{YT02})}
& \multicolumn{1}{|c|}{\F \B (Thm \ref{thm:counterex4p})}
\\
\cline{1-1} \cline{4-8}
\multicolumn{1}{|c|}{\F \USS}
& \multicolumn{2}{|c|}{\F } % implied by YT
& \multicolumn{5}{|c|}{\F \G (Thm \ref{thm:uniq})}
\\
\hline
\end{tabular}\\[1em]
\caption{\label{table:summary} Results summary: The impact of
  \USS/\SSS/\SS\ on weak acyclicity (with $M > 2$): \G\ marks classes with
  guaranteed weak acyclicity, even under best response; \B\ marks classes
  with examples that are not weakly acyclic even under
  better response. ${}^*$: only for strict games}
\end{center}
\vspace{-2em}
}
\end{table}

\subsection{Other Related Work}

Weak acyclicity has been specifically addressed in a handful of
specially-struc\-tured games: in an applied setting, BGP with backup
routing~\cite{ES09}, in a game-theoretical setting, games with
``strategic complementarities''~\cite{FM01,KTY05} (a supermodularity
condition on lattice-structured strategy sets), and in an algorithmic
setting, in several kinds of succinct games~\cite{MS09}.
Milchtaich~\cite{Mil96} studied Rosenthal's congestion
games~\cite{Ros73} and proved that, in interesting cases, such games
are weakly acyclic even if the payoff functions (utilities) are not
universal but player-specific.  Marden et al.~\cite{MAS07} formulated
the cooperative-control-theoretic consensus problem as a potential
game (implying that it is weakly acyclic); they also defined and
investigated a time-varying version of weak acyclicity.

\subsection{Outline of Paper}

In the following, we recall the relevant concepts and definitions
in Section~\ref{sec:defs}, present our sufficient condition for weak acyclicity in
Section~\ref{sec:suff}, and our characterization of weak acyclicity implications in
Section~\ref{sec:neg}.

\section{Weakly acyclic games and subgame stability}\label{sec:defs}

We use standard game-theoretic notation. Let $\Gamma$ be a
\emph{normal-form game} with $n$ players $1,\ldots,n$. We denote by
$S_i$ be the \emph{strategy space} of the $i^\mathrm{th}$ player. Let
$S=S_1\times \ldots \times S_n$, and let $S_{-i}=S_1\times
\ldots\times S_{i-1}\times S_{i+1}\times \ldots\times S_n$ be the
cartesian product of all strategy spaces but $S_i$. Each player $i$
has a \emph{utility function} $u_i$ that specifies $i$'s
\emph{payoff} in every strategy-profile of the players. For each
strategy $s_i\in S_i$, and every $(n-1)$-tuple of strategies
$s_{-i}\in S_{-i}$, we denote by $u_i(s_i,s_{-i})$ the utility
of the strategy profile in which player $i$ plays $s_i$ and
all other players play their strategies in $s_{-i}$.  We will make use of the following definitions.

\begin{definition} [better-response strategies]
A strategy $s'_i\in S_i$ is a \emph{better-response} of player $i$ to a
strategy profile $(s_i,s_{-i})$ if
$u_i(s'_i,s_{-i})>u_i(s_i,s_{-i})$.
\end{definition}

\begin{definition} [best-response strategies]
A strategy $s_i\in S_i$ is a \emph{best response} of player $i$ to a
strategy profile $s_{-i}\in S_{-i}$ of the other players if $s_i \in
\mathrm{argmax}_{s'_i\in S_i} u_i(s'_i,s_{-i})$
\end{definition}

\begin{definition} [pure Nash equilibria]
A strategy profile $s$ is a \emph{pure Nash equilibrium} if, for every player $i$, $s_i$ is a best response of $i$ to $s_{-i}$.
\end{definition}

\begin{definition} [better- and best-response improvement paths]
A \emph{better-response (best-response) improvement path} in a game $\Gamma$ is
a sequence of strategy profiles $s^1,\ldots,s^k$ such that for every
$j\in [k-1]$ (1) $s^j$ and $s^{j+1}$ only differ in the strategy of a
single player $i$ and (2) $i$'s strategy in $s^{j+1}$ is a better response
to $s^j$ (best response to $s^j_{-i}$ and $u_i(s^{j+1}_i,s^j_{-i})>u_i(s^j_i,s^j_{-i})$).  The \emph{better-response dynamics (best-response dynamics) graph} for $\Gamma$ is the graph on the strategy profiles in $\Gamma$ whose edges are the better-response (best-response) improvement paths of length 1.
\end{definition}

We will use $\Delta R_\Gamma (s)$ and $BR_\Gamma (s)$ to denote the
set of all states reachable by, respectively, better and best responses
when starting from $s$ in $\Gamma$.

We are now ready to define weakly acyclic games \cite{PeytonYoung93}. Informally,
a game is weakly acyclic if a pure Nash equilibrium can be reached from any initial strategy profile
via a better-response improvement path.

\begin{definition} [weakly acyclic games]
A game $\Gamma$ is \emph{weakly acyclic} if, from every strategy profile
$s$, there is a better-response improvement path $s^1\ldots,s^k$ such
that $s^1=s$, and $s^k$ is a pure Nash equilibrium in $\Gamma$. (I.e.,
for each $s$, there's a pure Nash equilibrium in $\Delta R_\Gamma (s)$.)
\end{definition}

We also coin a parallel definition based on best-response dynamics.

\begin{definition} [weak acyclicity under best response]
A game $\Gamma$ is \emph{weakly a\-cy\-clic under best response} if, from every
strategy profile $s$, there is a best-response improvement path $s^1\ldots,s^k$ such
that $s^1=s$ and $s^k$ is a pure Nash equilibrium in $\Gamma$. (I.e.,
for each $s$, there's a pure Nash equilibrium in $BR_\Gamma (s)$.)
\end{definition}

Weak acyclicity of either kind is equivalent to requiring that, under
the respective dynamics, the game has no ``non-trivial'' sink
equilibria~\cite{GMV05,FP08}, \ie, sink equilibria containing more
than one strategy profile. Conventionally, sink equilibria are defined
with respect to best-response dynamics, but the original definition by
Goemans et al.~\cite{GMV05} takes into account better-response dynamics as well.

The following follows easily from definitions:
\begin{claim}
If a game is weakly acyclic under best response then it is weakly acyclic.
\end{claim}
\begin{proof}
If $\Gamma$ is weakly acyclic under best response, the paths to equilibrium
from each edge will still be there if we augment the state space with
additional better-response transitions.
On the other hand, the game in Figure~\ref{fig:betternotbest}, mentioned in, \eg,~\cite{MYAS07}, is weakly
acyclic, but it is not weakly acyclic under best response.
\end{proof}

\begin{figure}[ht]
{ 
\ifdim \pgfversion pt < 2 pt \else
\tikzstyle{every picture}+=[remember picture]
\fi
\begin{center}
\newcommand{\N}[2]{\ifdim \pgfversion pt < 2 pt #2 \else
\tikz[baseline]{\node[anchor=base] (#1){#2};} \fi}
\begin{tabular}{c|ccc|}
  \cline{2-4}
  &
  $H$ & $T$ & $X$ \\
  \hline % \hline
  \multicolumn{1}{|c|}{$H$} &
  % dXXXX = dummy nodes that don't need arrows, but useful for making
  % the column widths match up
  \F \N{mp00}{$2,0$}          & \F \N{mp01}{$0,2$} & \F \N{mp02}{$0,0$}
  \\
  %\cline{3-11}
  %
  \multicolumn{1}{|c|}{$T$} &
  \F \N{mp10}{$0,2$} & \F \N{mp11}{$2,0$}           & \F \N{mp12Q}{$0,0$}
  \\
  %\cline{3-11}
  %
  \multicolumn{1}{|c|}{$X$} &
  \F \N{mp20}{$0,0$} & \F \N{mp21}{$1,0$}           & \F \N{mp22}{$3,3$}
  \\
  \hline %\hline
  %
  %%%%%%%%%%%%%%%%%%%%%%%%%%%%%%%%%%%%%%%%%%%%%%%%%%%%%%%%%%%%%%%%
  %
\end{tabular}
\end{center}
\ifdim \pgfversion pt < 2 pt
\else
\begin{tikzpicture}[overlay,>=latex]
\definecolor{edgecol}{rgb}{0.05,0.15,0.25}
\begin{scope}[color=edgecol,densely dashed]
\path[->] ($(mp01.base)+(.2,0)$) edge [out=0,in= 0,looseness=1.0] ($(mp21.base)+(.2,.1)$);
\path[->] ($(mp21.base)+(.2,0)$) edge [out=0,in= 180,looseness=1.0] ($(mp22.base)+(-.2,0)$);
\end{scope}
\begin{scope}[color=edgecol]
\path[->] ($(mp01.center)+(0,-.05)$) edge [out=-90,in= 90,looseness=1.0] ($(mp11.center)+(0,.1)$);
\path[->] ($(mp11.center)+(-.2,0)$) edge [out=180,in= 0,looseness=1.0] ($(mp10.center)+(.2,0)$);
\path[->] ($(mp10.center)+(0,.1)$) edge [out=90,in= -90,looseness=1.0] ($(mp00.center)+(0,-.05)$);
\path[->] ($(mp00.center)+(.2,0)$) edge [out=0,in= 180,looseness=1.0] ($(mp01.center)+(-.2,0)$);

\end{scope}
\end{tikzpicture}
\fi
}
\vspace{-1em}
\caption{\label{fig:betternotbest} Matching pennies with a ``better-response'' escape route (dashed arrows), but a best response persistent cycle (solid arrows)~\cite{MYAS07}.}
\vspace{-1em}
\end{figure}

Curiously, all of our results apply both to weak acyclicity in its
conventional better-response sense and to weak acyclicity under best
response. Thus, unlike weak acyclicity itself, the conditions 
presented in this paper
are ``agnostic'' to the better-/best-response distinction (like the notion of pure Nash equilibria itself).

We now present the notion of subgame stability.

\begin{definition} [subgames]
A \emph{subgame} of a game $\Gamma$ is a game $\Gamma'$ obtained from $\Gamma$ via
the removal of players' strategies.
\end{definition}

\begin{definition} [subgame stability]
\emph{Subgame stability} is said to hold for a game $\Gamma$ if every
subgame of $\Gamma$ has a pure Nash equilibrium. We use \SS\ to
denote the class of subgame stable games.
\end{definition}

\begin{definition} [unique subgame stability]
\emph{Unique subgame stability} is said to hold for a game $\Gamma$ if
every subgame of $\Gamma$ has a unique pure Nash equilibrium. We use
\USS\ to denote the class of such games.
\end{definition}

We will also consider games in which no player has two or more equally
good responses to any fixed set of strategies played by the other
players. Following, e.g., \cite{NSZ08}, we define \emph{strict games} as
follows.

\begin{definition}[strict game]
A game $\Gamma$ is \emph{strict} if, for any two distinct strategy profiles $s = (s_1,\ldots,s_n)$ and $s'=(s'_1,\ldots,s'_n)$ such that there is some $j\in[n]$ for which $s'= (s'_j, s_{-j})$ (\ie, $s$ and $s'$ differ only in $j$'s strategy), then $u_j(s)\neq u_j(s')$.
\end{definition}

\begin{definition} [SSS]
We use \SSS\ to denote the class of games that are both strict and
subgame stable.
\end{definition}

It's easy to connect unique subgame stability and strictness. To do
so, we use the next definition, which will also play a role in our
main proofs.

\begin{definition}[subgame spanned by profiles]
For game $\Gamma$ with $n$ players and profiles $s^1,\ldots,s^k$ in $\Gamma$, the \emph{subgame spanned by  $s^1,\ldots,s^k$} is the subgame $\Gamma'$ of $\Gamma$ in which the strategy space for player $i$ is $S'_i = \{s^j_i | 1\leq j\leq k\}$.
\end{definition}

\begin{claim}
The categories $\USS$, $\SSS$, and $\SS$ form a hierarchy:
$\USS \subset \SSS \subset \SS$
\end{claim}
\begin{proof}
$\SSS \subset \SS$ by definition. To see that $\USS\subset \SSS$ observe the following. If a game is not strict, there are $s_j,s'_j\in S_j$ and $s_{-j}$ such
that $u_j(s_j,s_{-j}) = u_j(s'_j,s_{-j})$.  Both strategy profiles in
the subgame spanned by $(s_j,s_{-j})$ and $(s'_j,s_{-j})$ are pure
Nash equilibria, violating unique subgame stability.
\end{proof}

\section{Sufficient condition for weak acyclicity with $n$ players}\label{sec:suff}
When is weak acyclicity guaranteed in $n$-player games for $n\geq 3$?
We prove that the existence of a \emph{unique} pure Nash equilibrium in
every subgame implies weak acyclicity. We note that this is not true when subgames
can contain multiple pure Nash equilibria~\cite{YT02}. Thus, while at first glance, introducing extra
equilibria might seem like it would make it harder to get ``stuck'' in a
non-trivial component of the state space with no ``escape path'' to an
equilibrium, this intuition is false; allowing extra pure Nash equilibria in subgames
actually enables the existence of non-trivial sinks.

\newcommand{\profile}[1]{\ensuremath{#1}}

\begin{theorem}\label{thm:uniq}
Every game $\Gamma$ that has a \emph{unique} pure Nash equilibrium
in every subgame $\Gamma'\subseteq\Gamma$ is weakly acyclic
under best-response (as are all of its subgames).
\end{theorem}

The proof of this theorem uses the following technical lemma:

\begin{lemma}\label{lma:br}
If $\profile{s}$ is a strategy profile in $\Gamma$, and $\Gamma'$ is
the subgame of $\Gamma$ spanned by $BR_\Gamma (\profile{s})$, then any
best-response improvement path $s,s^1,\ldots,s^k$ in $\Gamma'$ that starts at $s$ is also a
best-response improvement path in $\Gamma$.
\end{lemma}
\begin{proof}
We proceed by induction on the length of the path. The base case is tautological.  Inductively, assume $s,\ldots,s_i$ is a best-response improvement path in
$\Gamma$. The strategy $s^{i+1}$ is a best response to $s^i$ in $\Gamma'$
by some player $j$. This guarantees that $s^i$ is not a best response by
$j$ to $s^i_{-j}$ in $\Gamma'$, let alone in $\Gamma$, so
$\Gamma'\supseteq BR_\Gamma (s)\supseteq BR_\Gamma (s^i)$ must contain
a best-response $\hat{s}^i_j$ to $s^i_{-j}$ in $\Gamma$, and because
$s^{i+1}_{j}$ is a best-response in $\Gamma'$, we are guaranteed
that $u_j(\hat{s}^i_j,s^i_{-j})=u_j(s^{i+1})$, so $s^{i+1}$ must be a
best-response in $\Gamma$.
\end{proof}
We may now prove Theorem~\ref{thm:uniq}.

\begin{proof}[Proof of Theorem~\ref{thm:uniq}]
  To prove Theorem \ref{thm:uniq}, assume that $\Gamma$ is a game
  satisfying the hypotheses of the
  theorem, and for a subgame $\Delta\subseteq \Gamma$, denote by
  $\profile{s}_\Delta$ the unique pure Nash equilibrium in $\Delta$.
  We will proceed by induction up the semilattice of subgames of
  $\Gamma$. The base cases are trivial: any $1\times \cdots \times 1$
  subgame is weakly acyclic for lack of any transitions. Suppose that
  for some subgame $\Gamma'$ of game $\Gamma$ we know that every
  strict subgame $\Gamma'' \subsetneq \Gamma'$ is weakly acyclic.

  Suppose that $\Gamma'$ is not weakly acyclic: it has a state
  $\profile{s}$ from which its unique pure Nash equilibrium
  $\profile{s}_{\Gamma'}$ cannot be reached by best responses.
  Let $\Gamma''$ be the game
  spanned by $BR(\profile{s})$. Consider the cases of
  (i) $\profile{s}_{\Gamma'} \in \Gamma''$ and (ii)
  $\profile{s}_{\Gamma'} \notin \Gamma''$:

  \emph{Case (i)}: $\profile{s}_{\Gamma'} \in \Gamma''$. This requires
  that, for an arbitrary player $j$ with more than 1 strategy in
  $\Gamma'$, there be a best-response improvement path from $s$ to some
  profile $\hat{s}$ where $j$ plays the same strategy as it does in
  $s_{\Gamma''}$. Take one such $j$, and let $\Gamma^j$ be the subgame
  of $\Gamma'$ where $j$ is restricted to playing $\hat{s}_j$
  only. Because $s_{\Gamma'}$ is in $\Gamma^j$, the inductive hypothesis
  guarantees a best-response improvement path in $\Gamma^j$ from $\hat{s}$ to
  $s_{\Gamma'}$. By construction, that path must only involve
  best responses by players other than $j$, who have the same strategy
  options in $\Gamma^j$ as they did in $\Gamma'$, so that path is also
  a best-response improvement path in $\Gamma'$, assuring a best-response improvement path in
  $\Gamma'$ from $s$ to $s_{\Gamma'}$ via $\hat{s}$.

  \emph{Case (ii)}: $\profile{s}_{\Gamma'} \notin \Gamma''$. Then,
  $\Gamma''$'s unique pure equilibrium $\profile{s}_{\Gamma''}$ must
  be distinct from $\profile{s}_{\Gamma'}$. Because
  $\profile{s}_{\Gamma'}$ is the only pure equilibrium in $\Gamma'$,
  $\profile{s}_{\Gamma''}$ must have an outgoing best-response edge to
  some profile $\hat{s}$ in $\Gamma'$. But the inductive hypothesis
  ensures that $\profile{s}_{\Gamma''}\in BR_{\Gamma''}(s)$; by
  Lemma~\ref{lma:br}, $\profile{s}_{\Gamma''}\in BR_{\Gamma'}(s)$, which then
  ensures that $\hat{s}$ must also be in $BR_{\Gamma'}(s)$, and hence
  in $\Gamma''$, so $s_{\Gamma''}$ isn't an equilibrium in
  $\Gamma''$.
\end{proof}

\section{Characterizing the implications of subgame stability}\label{sec:neg}

Yamamori and Takahashi~\cite{YT02} established that in $2$-player games, subgame stability implies weak acyclicity, even under best response, yet this is not true in 3x3x3 games. We now present a complete characterization of when subgame stability is sufficient for weak
acyclicity, as a function of game size and strictness. Our next result shows that the
two-player theorem of \cite{YT02} is maximal:

\begin{theorem}
\label{thm:counterex-ss}
Subgame stability is not sufficient for weak acyclicity even in non-strict $2\times 2\times 2$ games.
\end{theorem}

\begin{proof}
The non-strict $2\times 2\times 2$ game in Fig.~\ref{fig:counterex-ss222} satisfies subgame stability but is not weakly acyclic---the states other than $(a_0,b_0,c_0)$ and $(a_1,b_1,c_1)$ form a non-trivial sink.
\end{proof}

\begin{figure}[ht]
\begin{center}
\newcommand{\B}{@{\hspace{0.6em}}}
\begin{tabular}{\B c \B | c \B c\B c \B| c \B c \B c \B|}
  \cline{2-7} %\cline{3-11}
  &
  \multicolumn{3}{|c|}{$c_0$} &
  \multicolumn{3}{|c|}{$c_1$} \\
  \cline{2-7}
  &&
  $b_0$ & $b_1$ &&
  $b_0$ & $b_1$
  \\
  \hline % \hline
  \multicolumn{1}{|c|}{$a_0$} && 2,2,2 & 1,2,2  &&  2,1,2 & 2,2,1 \\
  \multicolumn{1}{|c|}{$a_1$} && 2,2,1 & 2,1,2  &&  1,2,2 & 0,0,0 \\
  \hline
\end{tabular}%\hspace{1em}
\vspace{-0.5em}
\end{center}
%}
\caption{\label{fig:counterex-ss222} A non-strict $2\times 2\times 2$
  subgame-stable game with a non-trivial sink}
\end{figure}

However, if we require the games to be strict, subgame stability turns
out to be somewhat useful in 3-player games:

\begin{theorem}\label{thm:sss3p}
In any strict $2\times 2\times M$ or $2\times 3\times M$ game, subgame
stability implies weak acyclicity, even under best response.
\end{theorem}

\onlyshort{The proof of the theorem rests on
several technical lemmas:}

\onlylong{We will first need a couple of technical lemmas:}

\begin{lemma}
\label{lemma:nonashnbr}
In strict games, neither a pure Nash equilibrium nor strategy profiles
differing from it in only one player's action can be part of a
non-trivial sink of the best-response dynamics.
\end{lemma}
\begin{proof}
A pure Nash equilibrium always forms a 1-node sink. If the game is
strict, profiles differing by one player's action have to give that
one player a strictly lower payoff, requiring a best-response
transition to the equilibrium's sink. Any node connected to either
cannot be in a sink.
\end{proof}

\begin{lemma}
\label{lemma:sinkdoesntfit}
The profiles of a game that constitute a non-trivial sink of the
best-response dynamics cannot be all contained within a subgame
that is weakly acyclic under best-response.
\end{lemma}
\onlyshort{\begin{proof}[sketch] The lemma comes from considering, for a sink
of $\Gamma$ contained in a weakly acyclic subgame $\Gamma'$, a best
response path in $\Gamma'$ from the sink to an equilibrium. The first
transition on that path that is not a best response in $\Gamma$ (or,
in absence of such, the transition from the equilibrium of $\Gamma'$
that makes it a non-equilibrium in $\Gamma$), will have to lead out of
$\Gamma'$ but remain in the sink.
\end{proof}}
\onlylong{\begin{proof}
Consider such a non-trivial sink of game $\Gamma$ contained in such a
subgame $\Gamma'$. Take a profile $s$ in the sink, and consider the
path $P=\{s=s^0,s^1,\ldots,s^k\}$ of $\Gamma'$ best
responses that leads to $s^k$, an equilibrium of $\Gamma'$. This path
is guaranteed to exist because $\Gamma'$ is weakly acyclic under best
response. Consider the last profile in $P$, $s^a$, such that all profiles on
$P$ between $s$ and $s^a$ are in the sink.

If $s^a=s^k$ (i.e., if $P$ is entirely in the sink), there has to
be a best response transition in $\Gamma$ from $s^k$ to some
$s'$, because $s^k$ cannot be an equilibrium of $\Gamma$ and be
in a non-trivial sink. If $s'$ were in $\Gamma'$, the transition from
$s^k$ to $s'$ would have been a best response in $\Gamma'$,
too, contradicting $s^k$ being an equilibrium of $\Gamma'$---thus, $s'$ is not in $\Gamma'$, but is in the sink.

If $s^a\neq s^k$, the transition from $s^a$ to $s^{a+1}$, by some
player $i$, is a best response in $\Gamma'$, but not in $\Gamma$. So
$s^a_i$ is not $i$'s best response to $s^a_{-i}$, and thus there is a
best response by $i$ from $s^a$ to some $s'$ in $\Gamma$. Because $s^a$
to $s^{a+1}$ is \emph{not} a best response transition by in $\Gamma$,
$u_i(s')>u_i(s^{a+1})$, and because $s^{a+1}$ \emph{is} a best response
in $\Gamma'$, $s'$ must not be in $\Gamma'$---but because $s^a$ is in
the sink, so is $s'\notin \Gamma'$.
\end{proof}}

\onlyshort{We then consider the corner cases of 3-player, $2\times 2\times 2$
strict games, and 2-player, $2\times m$ games, where weak acyclicity requires
even less than subgame stability. The former result forms the base
case for Theorem \ref{thm:sss3p}, and both might also be of
independent interest.}

\onlylong{We now start with the corner cases of 3-player, $2\times 2\times 2$
games, and 2-player, $2\times m$ games, where weak acyclicity requires
even less than subgame stability. The former result forms the base
case for Theorem \ref{thm:sss3p}, and both might also be of
independent interest.}

\begin{lemma}
\label{lma:2m}
In any $2\times m$ game, if there is a pure Nash equilibrium, then the
game is weakly acyclic, even under best response.
\end{lemma}

\onlylong{\begin{proof}
In general $2\times m$ games with pure Nash equilibrium $(s^*,t^*)$, a
non-trivial best-response sink cannot consist of moves by just one
player.  Thus, the first player will play both of his strategies somewhere
in such a sink, including $s^*$, so there is some $(s^*,t')$ state in
the sink. If $t'$ is not a best response to $s^*$, there would be a
best-response transition to the equilibrium $(s^*,t^*)$, which
couldn't happen in a sink. If $t'$ is a best response to $s^*$, and
$s^*$ is a best response to $t'$, then $(s^*,t')$ is a Nash
equilibrium, which couldn't happen in a sink. Lastly, if $t'$ is a
best response to $s^*$, but $s^*$ is not a best response to $t'$,
there has to be some inbound best-response transition into $(s^*,t')$
from another profile in the sink, and that transition then has to
involve a move by player 2, from some other state $(s^*,t'')$,
guaranteeing that $t''$ is not a best response to $s^*$. Because $t^*$
has to also be a best response to $s^*$, there is then a best response
transition from $(s^*,t'')$ to the equilibrium $(s^*,t^*)$, concluding
the proof.
\end{proof}}

\begin{lemma}
\label{lma:222}
In any strict $2\times 2\times 2$ game,
if there is a pure Nash equilibrium, the game is weakly acyclic, even
under best response.
\end{lemma}
\begin{proof}
In strict $2\times 2\times 2$ games, Lemma \ref{lemma:nonashnbr}
leaves 4 other strategy profiles, with the possible best-response
transitions forming a star in the underlying undirected graph. Because
best-response links are antisymmetric ($s\to s'$ and $s'\to s$ cannot
both be best-response moves if they differ only in a single player's action), there can be no cycle among those 4
profiles, and thus no non-trivial sink components.
\end{proof}

\onlyshort{\begin{proof}[sketch of Theorem \ref{thm:sss3p}] The full
proof is long and technical, and is relegated to the full version of
the paper.

We treat the $2\times 2\times M$ case first. Naming the equilibrium
of the game $(a_0,b_0,c_0)$, Lemmas \ref{lemma:nonashnbr} and
\ref{lemma:sinkdoesntfit} guarantee that the sink must contain a
profile where player 3 plays $c_0$, yet the only such profile that can
be in the sink is $(a_1,b_1,c_0)$, the total degree of which in the
best-response directed graph is at most 1 (also by Lemma
\ref{lemma:nonashnbr}), which cannot happen for a node in a non-trivial
sink.

The $2\times 3\times M$ case is much more complex. The proof operates
inductively on $M$. From the inductive hypothesis, the $2\times
2\times M$ result, and Lemma \ref{lemma:sinkdoesntfit}, we get that
the \emph{smallest} $2\times 3\times M$ game $\Gamma$ that is not weakly
acyclic under best response must have a non-trivial sink spanning
$\Gamma$. Given such a sink, we then use Lemma \ref{lemma:nonashnbr}
and a similar result that excludes from the non-trivial sink any
profiles adjacent to the equilibrium of the $2\times 2\times M$
subgame that does not contain the global equilibrium. The proof
concludes by a detailed examination of the possible structures of such
a sink under all those constraints, which yield a contradiction in
every case.
\end{proof}}

\onlylong{\begin{proof}[Proof of Theorem~\ref{thm:sss3p}]
We treat the $2\times 2\times M$ and $2\times 3\times M$ cases
separately.

\noindent{\bf The $2\times 2\times M$ case}: With Lemma \ref{lma:222} as the
base case, assume, inductively, that the $2\times 2\times M$ claim
holds for all values of $M$ through some $M'-1$, and suppose some
$2\times 2\times M'$ game $\Gamma$, with strategy sets
$\{a_{0,1}\}$ (i.e., the set containing the strategies $a_0$ and $a_1$), $\{b_{0,1}\}$, and $\{c_{0,\ldots,M'-1}\}$, has a non-trivial
best-response sink $X$. Without loss of generality, let $(a_0,b_0,c_0)$ be an equilibrium of
$\Gamma$.

Lemma \ref{lemma:sinkdoesntfit} guarantees that $X$
is not contained in the subgame $\Gamma_{-c_0}$, where only strategy
$c_0$ is removed, leaving a strict, subgame stable $2\times 2\times
M'-1$ game, which is weakly acyclic under best response by the
inductive hypothesis. But the only profile using $c_0$ that is allowed
to be in $X$ after applying Lemma \ref{lemma:nonashnbr} is
$(a_1,b_1,c_0)$, from which the same lemma guarantees that only player
3 can make a best-response transition in $X$. Thus, it can have no inbound
best-response transitions by player 3, leaving no way for it to be
reached from the rest of $X$, which can thus not be a sink.

\noindent{\bf The $2\times 3\times M$ case}: The $2\times 3\times 2$ case is
isomorphic to the above. With that as the base case, assume,
inductively, that the $2\times 3\times M$ claim holds for all values
of $M$ up to some $M'-1$, and suppose that a $2\times 3\times M'$ game
$\Gamma$, with strategy sets $\{a_{0,1}\}$,$\{b_{0,1,2}\}$,and
$\{c_{0,\ldots,M'-1}\}$ has a non-trivial sink $X$. Without loss of generality, let
$(a_0,b_0,c_0)$ be a pure Nash equilibrium of $\Gamma$. The inductive
hypothesis and Lemma \ref{lemma:sinkdoesntfit} guarantee that $X$
spans $\Gamma$, and, in particular, that it has at least one node of form
$(*,*,c_0)$, and at least one of form $(*,b_0,*)$.

By Lemma \ref{lemma:nonashnbr}, the $(*,*,c_0)$ node has to be one of
the two nodes $(a_1,b_{1,2},c_0)$, and that node cannot have an
outbound best response by player 1. To be in a non-trivial sink, it
has to have an inbound and an outbound best response, one of which is
thus by player 3, and the other by player 2, ensuring that \emph{both}
of the two nodes $(a_1,b_{1,2},c_0)$ are in $X$. One of those will
then be player 2's best response to $(a_1,c_0)$; without loss of generality, let that one be
$(a_1,b_1,c_0)$. Then, the only inbound best response to lead to
$(a_1,b_2,c_0)$ is by player 3, and player 3 has to have an outbound
best response from $(a_1,b_1,c_0)$ to some $(a_1,b_1,c_x)$.

From $(a_1,b_1,c_x)$, if there is an outbound best response by player
2, it cannot be to $(a_1,b_2,c_x)$: otherwise, the subgame with strategies
$\{a_1\}$, $\{b_{1,2}\}$, and $\{c_{0,x}\}$ is isomorphic to Matching
Pennies. Player 2's best response would thus have to instead be to
$(a_1,b_0,c_x)$. From there, player 1 cannot have an outbound best
response by Lemma \ref{lemma:nonashnbr}, thus requiring player 3 to
have a best response to some $(a_1,b_0,c_y)$; from there, too, player
1 cannot have a best response by Lemma \ref{lemma:nonashnbr},
requiring a best response by player 2. But then, in the $1\times
3\times M'$ subgame formed by removing strategy $a_0$, for each of
player 2's strategies, player 3's best response is to a profile that
has an outbound best response by player 2, which precludes an
equilibrium.

Thus, from $(a_1,b_1,c_x)$, the sole possible outbound best response is by
player 1, to $(a_0,b_1,c_x)$.

\newcommand{\GB}{\ensuremath{\Gamma_{-b_0}}}
\newcommand{\NB}{\ensuremath{s^*}}

Consider now the $2\times 2\times M'$ subgame $\GB$ formed by taking
away strategy $b_0$, and let $\NB$ be its pure Nash equilibrium. If
$\NB$ is of form $(a_0,b_{1,2},c_0)$, that would require that it be
player 1's best response in $\Gamma$ to $(b_{1,2},c_0)$, thus putting
$\NB$ in the sink, in violation of Lemma \ref{lemma:nonashnbr}. The pure Nash equilibrium $\NB$
also cannot be of form $(a_1,b_{1,2},c_0)$: otherwise, it is in the
sink, and yet the only outbound best responses in $\Gamma$ must be
those not in $\Gamma'$, i.e., by player 2 to $(a_1,b_0,c_0)$, in
violation of Lemma \ref{lemma:nonashnbr}.  Thus, $\NB$ has player 3 playing a strategy other
than $c_0$, which is its best response to one of
$(a_1,b_{1,2})$. Player 3's best response to $(a_1,b_2)$ is
$c_0$, so that cannot be $\NB$. Player 3's best response to
$(a_1,b_1)$ is $c_x$, but there is an outbound best response by player
1 to $(a_0,b_1,c_x)$ from there, which is within $\GB$. Thus,
$\NB=(a_0,b_y,c_z)$, for $y\in \{1,2\}$.

Then, $\NB$ again cannot be in $X$, because the sole outbound transition
in $\Gamma$ could only be to $(a_0,b_0,*)$, violating Lemma
\ref{lemma:nonashnbr}. Neither $(a_1,b_y,c_z)$ nor any $(a_0,b_y,*)$
profile can be in $X$, either, because their best-response transition to
$\NB$ in $\Gamma'$ would also be a best response in $\Gamma$, putting
$\NB$ into $X$. If the one $(a_0,b_v,c_z)$ profile with $0\neq v\neq
y$ were in $X$, it has a best response to $\NB$ in $\Gamma'$, so it
has to have a best response by player 2 in $\Gamma$---but it cannot
be to $(a_0,b_0,c_z)$ by Lemma \ref{lemma:nonashnbr}, and cannot be to
$\NB$ because $\NB$ cannot be in $X$. Thus, much like in Lemma 2, no
profile differing in at most one player's strategy from $\NB$ can be
in $X$, either.

We can now show that nodes $(a_1,b_{0,v},c_z)$ are both in $X$, by an
argument symmetrical to that for $(a_1,b_{1,2},c_0)$.  The same
argument will yield that either $b_v$ or $b_0$ is the best response to
$(a_1,c_z)$, and that the other one of the two is a best response by
player 3. We finish by analyzing the cartesian product of those two
cases, and whether $v\in \{1,2\}$:

\textbf{Case:} $v=1$, $b_v$ is best response. The above argument will require that
the best response to $(a_1,b_1)$ by player 3, $c_x$, be neither $c_0$
nor $c_z$. If the outbound best response from $(a_1,b_1,c_x)$ is by
player 2, to $b_0$ or $b_2$, then either $\{a_1\}\times \{b_{0,1}\}
\times \{c_{z,x}\}$ or $\{a_1\} \times \{b_{1,2}\} \times
\{c_{0,x}\}$, respectively, form a subgame isomorphic to matching
pennies. On the other hand, suppose the outbound best response from
$(a_1,b_1,c_x)$ is by player 1, to $(a_0,b_1,c_x)$. Since
$(a_0,b_0,*)$ nodes and $(a_0,b_2,*)$ nodes may not be in $X$, the
only outbound response from there is by player 3, to some $c_w$, from
which the only outbound best response is by player 1 to
$(a_1,b_1,c_w)$, creating a matching pennies subgame with strategies
$\{a_{0,1}\} \times \{b_1\} \times \{c_{w,x}\}$.

\textbf{Case:} $v=1$, $b_0$ is best response. In this case, $c_z$ has to be the best response
to $(a_1,b_1)$ by player 3, requiring that $c_x=c_z$, but it was
established above that $(a_1,b_1,c_x)$ cannot have an outbound best
response by player 2.

\textbf{Case:} $v=2$, $b_v$ is best response. An argument symmetric to the $v=1$,
$b=0$ case will show that $c_0$, the requisite best response to
$(a_1,b_1)$ cannot have an outbound best response by player 2.

\textbf{Case:} $v=2$, $b_0$ is best response. This gives a contradiction, because it would require
both $c_z$ and $c_0$ to be the best response to $(a_1,b_2)$.

Thus, $\Gamma$ cannot have a non-trivial best-response sink.
\end{proof}}

Theorem~\ref{thm:sss3p} is maximal. All bigger sizes of 3-player games
admit subgame-stable 
examples that are not weakly acyclic:

\begin{theorem}
\label{thm:counterex3p}
In non-degenerate\footnote{Each player has 2 or more strategies}
strict 3-player games, the existence of pure Nash equilibria in every
subgame is insufficient to guarantee weak acyclicity, for any game
with at least 3 strategies for each player, and any game with at least
4 strategies for 2 of the players.
\end{theorem}

\begin{proof}
The first half of the theorem follows directly from a specific 
example game in \cite{YT02}. There, the strict
3-player, $3\times 3\times 3$ game in question is stated to
demonstrate that $\SSS$ does not imply weak acyclicity under best
response. Actually, their very same 
example is not even weakly
acyclic under better response.
\onlyshort{Here, we provide a $2\times 4\times 4$ 
example to establish the second half of the theorem, and a $3\times 3\times 3$
counterexample slightly cleaner than the one in \cite{YT02}.  Close
inspection of the games $\Gamma_{3,3,3}$ and $\Gamma_{4,4,2}$ shown in
Figure \ref{fig:counterex3p} reveals that these are not weakly acyclic
but are strict and subgame stable.}
\onlylong{Here, we examine a $2\times 4\times 4$ 
example to establish the second half of the theorem, and a $3\times 3\times 3$ 
example that is slightly cleaner than the one in \cite{YT02}, both
shown in Figure~\ref{fig:counterex3p}.}

\begin{figure}[tbh]
{ % tikz hack to overlay arrows on the table
\ifdim \pgfversion pt < 2 pt \else
\tikzstyle{every picture}+=[remember picture]
\fi
\begin{center}
\newcommand{\N}[2]{\ifdim \pgfversion pt < 2 pt #2 \else
\tikz[baseline]{\node[anchor=base] (#1){#2};} \fi}
\begin{tabular}{c|ccc|ccc|ccc|}
  \cline{2-10} %\cline{3-11}
  &
  \multicolumn{3}{c|}{$c_0$} &
  \multicolumn{3}{c|}{$c_1$} &
  \multicolumn{3}{c|}{$c_2$} \\
  \cline{2-10}
  &
  $b_0$ & $b_1$ & $b_2$ &
  $b_0$ & $b_1$ & $b_2$ &
  $b_0$ & $b_1$ & $b_2$
  \\
  \hline % \hline
  \multicolumn{1}{|c|}{$a_0$} &
  % dXXXX = dummy nodes that don't need arrows, but useful for making
  % the column widths match up
  \F $0,0,0$           & \F \N{n010}{$5,5,4$} & \F \N{n020}{$5,4,5$} &
  \F \N{n001}{$4,5,5$} & \F $0,1,1$           & \F \N{d021}{$0,2,1$} &
  \F \N{n002}{$5,5,4$} & \F \N{n012}{$5,4,5$} & \F \N{d022}{$0,2,2$}
  \\
  %\cline{3-11}
  %
  \multicolumn{1}{|c|}{$a_1$} &
  \F \N{n100}{$5,4,5$} & \F $1,1,0$           & \F \N{n120}{$4,5,5$} &
  \F $1,0,1$           & \F $5,5,5$           & \F $1,2,1$ &
  \F $1,0,2$           & \F $1,1,2$           & \F $1,2,2$
  \\
  %\cline{3-11}
  %
  \multicolumn{1}{|c|}{$a_2$} &
  \F \N{n200}{$4,5,5$} & \F $2,1,0$           & \F $2,2,0$ &
  \F \N{n201}{$5,5,4$} & \F $2,1,1$           & \F $2,2,1$ &
  \F $2,0,2$           & \F $2,1,2$           & \F $2,2,2$
  \\
  \hline %\hline
\end{tabular}
\\[1em]

\begin{tabular}{c|cccc|cccc|}
  \cline{2-9} %\cline{3-11}
  &
  \multicolumn{4}{c|}{$c_0$} &
  \multicolumn{4}{c|}{$c_1$} \\
  \cline{2-9}
  &
  $b_0$ & $b_1$ & $b_2$ & $b_3$ &
  $b_0$ & $b_1$ & $b_2$ & $b_3$
  \\
  \hline % \hline
  \multicolumn{1}{|c|}{$a_0$} &
  \F \N{m000}{$5,5,5$} & \F $0,1,0$ & \F          $0,2,0$  & \F      $0,3,0$ &
  \F \N{m001}{$5,5,4$} & \F $0,1,1$ & \F \N{m021}{$5,4,5$} & \F      $0,3,1$
  \\
  %\cline{3-11}
  %
  \multicolumn{1}{|c|}{$a_1$} &
  \F $1,0,0$       & \F $1,1,0$ & \F \N{m120}{$5,5,4$} & \F \N{m130}{$5,4,5$} &
  \F $1,0,1$       & \F $1,1,1$ & \F \N{m121}{$4,5,5$} & \F          $1,3,1$
  \\
  %\cline{3-11}
  %
  \multicolumn{1}{|c|}{$a_2$} &
  \F $2,0,0$       & \F \N{m210}{$5,4,5$} & \F $2,2,0$ & \F \N{m230}{$4,5,5$} &
  \F $2,0,1$       & \F $2,1,1$           & \F $2,2,1$ & \F          $2,3,1$
  \\
  %\cline{3-11}
  %
  \multicolumn{1}{|c|}{$a_3$} &
  \F \N{m300}{$5,4,5$} & \F \N{m310}{$4,5,5$} & \F $3,2,0$ & \F $3,3,0$ &
  \F $3,0,1$           & \F          $3,1,1$  & \F $3,2,1$ & \F $5,5,5$
  \\
  %\cline{3-11}
  %
  \hline %\hline
\end{tabular}
\end{center}
\ifdim \pgfversion pt < 2 pt
\else
\begin{tikzpicture}[overlay,>=latex]
\definecolor{edgecol}{rgb}{0.05,0.15,0.25}
\begin{scope}[color=edgecol,dashed]
\path[->] ($(n200.base)+(.2,0)$) edge [out= 100,in= 260,looseness=1.0] ($(n100.base)+(.2,0)$);
\path[->] ($(n100.base)+(.2,0)$) edge [out= 340,in= 225,in looseness=0.5] ($(n120.base)+(.2,0)$);
\path[->] ($(n120.base)+(.2,0)$) edge [out= 100,in= 260,looseness=1.0] ($(n020.base)+(.2,0)$);
\path[->] ($(n020.base)+(.2,0)$) edge [out= 210,in= 270,in looseness=1.0] ($(n010.center)+(.2,0)$);
\path[->] ($(n010.center)+(.2,0)$) edge [out=  90,in= 100,out looseness=0.2,in looseness=0.2] ($(n012.base)+(.2,0)$);
\path[->] ($(n012.base)+(.2,0)$) edge [out= 200,in= 335,looseness=1.0] ($(n002.base)+(.2,0)$);
\path[->] ($(n002.base)+(.2,0)$) edge [out= 200,in= 335,looseness=0.3] ($(n001.base)+(.2,0)$);
\path[->] ($(n001.base)+(.2,0)$) edge [out= 265,in=  95,looseness=1.0] ($(n201.base)+(.2,0)$);
\path[->] ($(n201.base)+(.2,0)$) edge [out= 200,in= 335,looseness=0.3] ($(n200.base)+(.2,0)$);
\path[->] ($(m001.center)+(.2,0)$) edge [out=  90,in=  40,out looseness=0.2,in looseness=0.2] ($(m000.center)+(.2,.1)$);
\path[->] ($(m000.center)+(.2,.1)$) edge [out= 265,in=  95,looseness=1.0] ($(m300.base)+(.2,0)$);
\path[->] ($(m300.base)+(.2,0)$) edge [out= 330,in= 200,looseness=0.5] ($(m310.base)+(.2,0)$);
\path[->] ($(m310.base)+(.2,0)$) edge [out= 100,in= 260,looseness=1.0] ($(m210.base)+(.2,0)$);
\path[->] ($(m210.base)+(.2,0)$) edge [out= 340,in= 225,looseness=0.3] ($(m230.base)+(.2,0)$);
\path[->] ($(m230.base)+(.2,0)$) edge [out= 100,in= 260,looseness=1.0] ($(m130.base)+(.2,0)$);
\path[->] ($(m130.base)+(.2,0)$) edge [out=  90,in=  80,looseness=0.6] ($(m120.center)+(.2,0)$);
\path[->] ($(m120.center)+(.2,0)$) edge [out= 310,in= 210,looseness=0.25] ($(m121.base)+(.2,0)$);
\path[->] ($(m121.base)+(.2,0)$) edge [out= 100,in= 260,looseness=0.5] ($(m021.base)+(.2,0)$);
\path[->] ($(m021.base)+(.2,0)$) edge [out= 210,in= 340,looseness=0.5] ($(m001.base)+(.2,0)$);
\end{scope}
\end{tikzpicture}
\fi
}
\vspace{-1em}
\caption{\label{fig:counterex3p}
3-player strict subgame stable games that are not weakly
  acyclic, even under better-response dynamics}
\vspace{-1em}
\end{figure}

In each of these three player games, there is a pure Nash equilibrium
in the full game, $s^*=(a_1,b_1,c_1)$ in $\Gamma_{3,3,3}$, and
$s^*=(a_3,b_3,c_1)$ in $\Gamma_{4,4,2}$, with utility 5 for each of
the players. In both, there is a cycle $C$, every profile in which
differs from $s^*$ in at least 2 players' strategies. Any profile
$(a_i,b_j,c_k)$ that's neither $s^*$ nor in $C$ yields utilities
$(i,j,k)$. With utilities in $C$ always in $\{4,5\}$, there is never an
incentive for anyone to unilaterally leave the cycle $C$, forming a
``sheath'' of low-utility states separating $C$ from the rest of the
game, particularly $s^*$. Thus $C$ is a persistent cycle. By
construction, the game is strict and at each state in $C$ there is a
unique player who has a better response to the current state.

Consider any subgame $\Gamma'$ of either game. If $\Gamma'$ contains
$s^*$, $s^*$ is a pure Nash equilibrium of $\Gamma'$ as well.

Suppose $\Gamma'$ is not the full game. In the course of cycling
through $C$, each strategy of each player is used at least once. Thus,
$\Gamma'$ cannot contain all of $C$. If it has at least some states of
$C$, pick one state that is in $\Gamma'$, and follow the edges of $C$
until you get to a state whose sole outbound better-response move has been
``broken'' by the better-response strategy being removed in
$\Gamma'$. This process will terminate, because $C$ is a simple cycle in
$\Gamma$ that had at least one node missing in $\Gamma'$.  The sole
player that had an incentive to move in that state in $\Gamma$ now no
longer has that option, and if he has any other strategy, the
resulting state cannot be in $C$, because $C$ never uses more than 2
strategies of any player $i$ in combination with any fixed
$s_{-i}$. Thus, any other strategy is not an improvement for that
player, either, and this new state is thus a pure Nash equilibrium in
$\Gamma'$.

Lastly, if $\Gamma'$ contains neither $s^*$ nor any nodes of $C$,
taking the highest-index strategy for each player yields a profile
that has to be a pure Nash equilibrium, because the utilities of non-$C$,
non-$s^*$ profiles are just $(i,j,k)$.

Thus, every subgame is guaranteed to have a pure Nash equilibrium, and, due
to $C$, both games are not weakly acyclic. The theorem holds for games
with more strategies by padding the 
examples given above.
\end{proof}

With 4 or more players, a more mechanistic approach produces
analogous examples even with just 2 strategies per player:

\begin{theorem}
\label{thm:counterex4p}
In a strict $n$-player game for an arbitrary $n\geq 4$, the existence of
pure Nash equilibria in every subgame is insufficient to guarantee
weak acyclicity, even with only 2 strategies per player.
\end{theorem}
\begin{proof}
For strategy profiles in $\{0,1\}^n$, using indices mod $n$, set the utilities to:
\begin{equation*}
\vec{u}(\vec{s}) = \begin{cases}
(4,\ldots,4) & \text{at $\vec{s}=(1,\ldots,1)$} \\
(3,\ldots,3,\underset{\text{$i$'th}}{2},3,\ldots,3) & \text{when
    $s_{i-1}=s_i=1$, $s_{-(i-1,i)}=0$}\\
(3,\ldots,3,\underset{\text{$i+1$'th}}{2},3,\ldots,3) & \text{when
    $s_{i}=1$, $s_{-i}=0$}\\
\vec{s} & \text{else (for the ``sheath'')}.
\end{cases}
\end{equation*}
Similarly to Theorem \ref{thm:counterex3p}, this plants a global pure Nash
equilibrium at $(1,\ldots,1)$, and creates a ``fragile'' better-response
cycle. Here, the cycle alternates between profiles with edit distance
$n-1$ and $n-2$ from the global pure Nash equilibrium. At every point of
the cycle, the only non-sheath profiles 1 step away are its
predecessor and successor on the cycle, so the cycle is
persistent. Because each profile with edit distance $n-1$ from the
equilibrium is covered, removing any player's 1 strategy breaks the
cycle, thus guaranteeing a pure Nash equilibrium in every subgame by the
same reasoning as above.
\end{proof}

We note that the fixed-size examples that demonstrate the negative results above---in Theorems
\ref{thm:counterex-ss}, \ref{thm:counterex3p}, and \ref{thm:counterex4p}---easily extend to games with extra strategies for some or all players, or with extra players,
by ``padding'' the added part of the payoff table with negative,
unique values that, for the added profiles, make payoffs independent
of the other players, such as, e.g., $u_i(s)=-s_i$. This preserves
$\SS$, $\SSS$, and $\USS$ properties without changing weak
acyclicity. Thus, this completes our classification of weak acyclicity
under the three subgame-based properties, as shown in Table
\ref{table:summary}.

\section{Concluding remarks}\label{sec:conc}

The connection between weak acyclicity and unique subgame stability
that we present is surprising, but not immediately practicable: in
most succinct game representations, there is no reason to believe that
checking unique subgame stability will be tractable in many general
settings. In a complexity-theoretic sense, $\USS$ is
\emph{closer} to tractability than weak acyclicity: Any reasonable
game representation will have some ``reasonable'' representation of
subgames, i.e., one in which \emph{checking} whether a state is a pure
Nash equilibrium is tractable, which puts unique subgame stability in a
substantially easier complexity class, $\Pi_3 P$, than the class PSPACE
for which weak acyclicity is complete in many games.

We leave open the important question of finding efficient algorithms
for checking unique subgame stability, which may well be feasible in
particular classes of games. Also open and relevant, of course, is the
question of more broadly applicable and tractable conditions for weak
acyclicity. In particular, there may well be other levels of the
subgame stability hierarchy between $\SSS$ and $\USS$ that could give
us weak acyclicity in broader classes of games.

\end{document}